\documentclass[preprint,10pt]{elsarticle}

\usepackage{lineno,hyperref}
\usepackage{amsmath}
\usepackage{subfigure}
\usepackage{caption}
\usepackage{xcolor}
\usepackage{amssymb}
\usepackage{lipsum}
\usepackage{amsfonts}
\usepackage{graphicx}
\usepackage{epstopdf}
\usepackage{algorithmic}
\usepackage{amsmath}
\usepackage{bm}
\usepackage{mathrsfs}
\usepackage{amsthm}

\newtheorem{theorem}{Theorem}
\newtheorem{lemma}[theorem]{Lemma}

\newtheorem{remark}{Remark}

\newtheorem{problem}[theorem]{Problem} 
\newtheorem{proposition}[theorem]{Proposition}  
\modulolinenumbers[5]

\journal{Journal of COMPUT METHOD APPL M}

\biboptions{numbers,sort&compress}
\begin{document}

\begin{frontmatter}

\title{An Energy Stable One-Field Fictitious Domain Method for Fluid-Structure Interactions} %\tnoteref{mytitlenote}}
%\tnotetext[mytitlenote]{Fully documented templates are available in the elsarticle package on \href{http://www.ctan.org/tex-archive/macros/latex/contrib/elsarticle}{CTAN}.}

%% Group authors per affiliation:
%\author{Elsevier\fnref{myfootnote}}
%\address{Radarweg 29, Amsterdam}
%\fntext[myfootnote]{Since 1880.}

%% or include affiliations in footnotes:
\author[]{Yongxing Wang\corref{mycorrespondingauthor}}
\cortext[mycorrespondingauthor]{Corresponding author}
\ead{jungsirwang@gmail.com/scywa@leeds.ac.uk}
%\ead[url]{www.elsevier.com}

\author[]{Peter K. Jimack}
\author[]{Mark A. Walkley}

\address{School of Computing, University of Leeds, Leeds, UK, LS2 9JT}

\begin{abstract}
In this article, the energy stability of a one-field fictitious domain method is proved and validated by numerical tests in two and three dimensions. The distinguishing feature of this method is that it only solves for one velocity field for the whole fluid-structure domain; the interactions remain decoupled until solving the final linear algebraic equations. To achieve this the finite element procedures are carried out separately on two different meshes for the fluid and solid respectively, and the assembly of the final linear system brings the fluid and solid parts together via an isoparametric interpolation matrix between the two meshes. The weak formulations are introduced in the continuous case and after discretization in time. Then the stability is analyzed through an energy estimate. Finally, numerical examples are presented to validate the energy stability properties.
\end{abstract}


\end{frontmatter}

\linenumbers

%%%%%%%%%%%%%%%%%%%%%%%%%%%%%%%%%%%%
\section{Introduction}
\label{sec:introduction}
Three major questions arise when considering a finite element method for the problem of Fluid-Structure Interactions (FSI): (1) what kind of meshes are used (interface fitted or unfitted); (2) how to couple the fluid-structure interactions (monolithic/fully-coupled or partitioned/segregated); (3) what variables are solved (velocity and/or displacement). Combinations of the answers of these questions lead to different types of numerical method. For example, \cite{Degroote_2009, K_ttler_2008}  solve for fluid velocity and solid displacement sequentially (partitioned/segregated) using an Arbitrary Lagrangian-Eulerian (ALE) fitted mesh, whereas \cite{Heil_2004, Heil_2008, Muddle_2012} use an ALE fitted mesh to solve for fluid velocity and solid displacement simultaneously (monolithic/fully-coupled) with a Lagrange Multiplier to enforce the continuity of velocity/displacement on the interface. The Immersed Finite Element Method (IFEM) \cite{Boffi_2015, peskin2002immersed, Wang_2011, Wang_2009,Wang_2013,Zhang_2007,zhang2004immersed} and the Fictitious Domain Method (FDM) \cite{baaijens2001fictitious,Boffi_2016,Glowinski_2001,Hesch_2014,Kadapa_2016,Yu_2005} use two meshes to represent the fluid and solid separately. Although IFEM could be monolithic \cite{Boffi_2015}, the classical IFEM only solves for velocity, while the solid information is arranged on the right-hand side of the fluid equation as a prescribed force term. Although the FDM may be partitioned \cite{Yu_2005}, usually the FDM approach solves for both velocity in the whole domain (fluid plus solid) and displacement of the solid simultaneously via a distributed Lagrange multiplier (DLM) to enforce the consistency of velocity/displacement in the overlapped solid domain. In the case of one-field and monolithic numerical methods for FSI problems, \cite{Auricchio_2014} introduces a 1D model using a one-field FDM formulation based on two meshes, and \cite{Hecht_2017,Pironneau_2016} introduces an energy stable monolithic method (in 2D) based on one Eulerian mesh and discrete remeshing. 

In a previous study \cite{Wang_2017}, we present a one-field monolithic fictitious domain method (subsequently referred to as the one-field FDM) which has the following main features: (1) only one velocity field is solved in the whole domain, based upon the use of an appropriate $L^2$ projection; (2) the fluid and solid equations are solved monolithically. Our motivation for proposing the one-field FDM is based on comparing its features with those of existing numerical schemes. Compared with IFEM the classical IFEM does not solve the solid equation \cite{Wang_2011, Wang_2009,Wang_2013,Zhang_2007,zhang2004immersed}. Instead, the solid information is arranged on the right-hand side of the fluid equation as a prescribed force. The one-field FDM solves the solid equation together with the fluid equation in one discretized linear algebraic system. The similarity is that both methods only solve for velocity and pressure fields (no solid displacement). DLM/FDM methods \cite{baaijens2001fictitious,Boffi_2016,Glowinski_2001,Hesch_2014,Kadapa_2016,Yu_2005} solve the solid equation, but for a displacement field, and couple this displacement with the velocity of the fictitious fluid via a Lagrange multiplier. This leads to a large discretized linear algebra system. The one-field FDM rewrites the solid equation in terms of a velocity variable and couples the fictitious fluid through a finite element interpolation. Monolithic Eulerian methods \cite{Hecht_2017, Pironneau_2016} also express the solid equation in terms of velocity, and the fluid and solid are coupled naturally on an interface-fitted mesh. The one-field FDM uses two meshes to represent the fluid and solid respectively. Consequently, before discretization in space, these two methods have many similarities, the advantage of the one-field FDM being that interface fitting is not required.

The main developments in this paper, following from \cite{Wang_2017}, are as follows. The energy preserving property in the continuous case is proved. The energy nonincreasing property after time discretization is proved, and the same property is also proved after spatial discretization. The implementation in this paper is based on an ${\bf F}$-scheme, i.e., the solid deformation tensor ${\bf F}$ is updated (see section \ref{sec_weak_form_time_discretization}), while the previous paper uses a $\bm{\sigma}$-scheme (see equation (29) in \cite{Wang_2017}). The advantage of this ${\bf F}$-scheme is that the integral is expressed in the reference domain so that it becomes linear for the neo-Hookean solid model (see equation (\ref{weak_form1_discretization_backward_Euler})), which however is nonlinear if expressed in the current domain for the $\bm{\sigma}$-scheme (see equation (29) in \cite{Wang_2017}). The methodology and analysis is demonstrated to extend to the three-dimensional case.

The paper is organized as follows. Control equations and weak formulation are introduced in section \ref{control_equations} and \ref{sec:wfotcl} respectively. The time discretized weak form is then presented in section \ref{sec_weak_form_time_discretization}. Stability of the proposed scheme is analyzed in section \ref{stability_analysis}. Space discretization is discussed in section \ref{sec:weak_form_space_discretization}. Numerical examples are given in section \ref{sec:numerical_exs}, and conclusions are presented in section \ref{sec:conclusions}.

%%%%%%%%%%%%%%%%%%%%%%%%%%%%%%%%%%%%%%
\section{Control equations}
\label{control_equations}
In the following context, $\Omega_t^f\subset\mathbb{R}^d$ and $\Omega_t^s\subset\mathbb{R}^d$ with $d=2,3$ denote the fluid and solid domain respectively which are time dependent regions as shown in Figure \ref{fig1:Schematic diagram of FSI}. $\Omega=\Omega_t^f \cup \Omega_t^s $ is a fixed domain (with outer boundary $\Gamma$) and $\Gamma_t=\partial\Omega_t^f\cap\partial\Omega_t^s$ is the moving interface between fluid and solid. We denote by ${\bf X}$ the reference (material) coordinates of the solid, by ${\bf x}={\bf x}(\cdot,t)$ the current coordinates of the solid, and by ${\bf x}_0$ the initial coordinates of the solid.

\begin{figure}[h!]
	\centering
	\includegraphics[width=2.5in,angle=0]{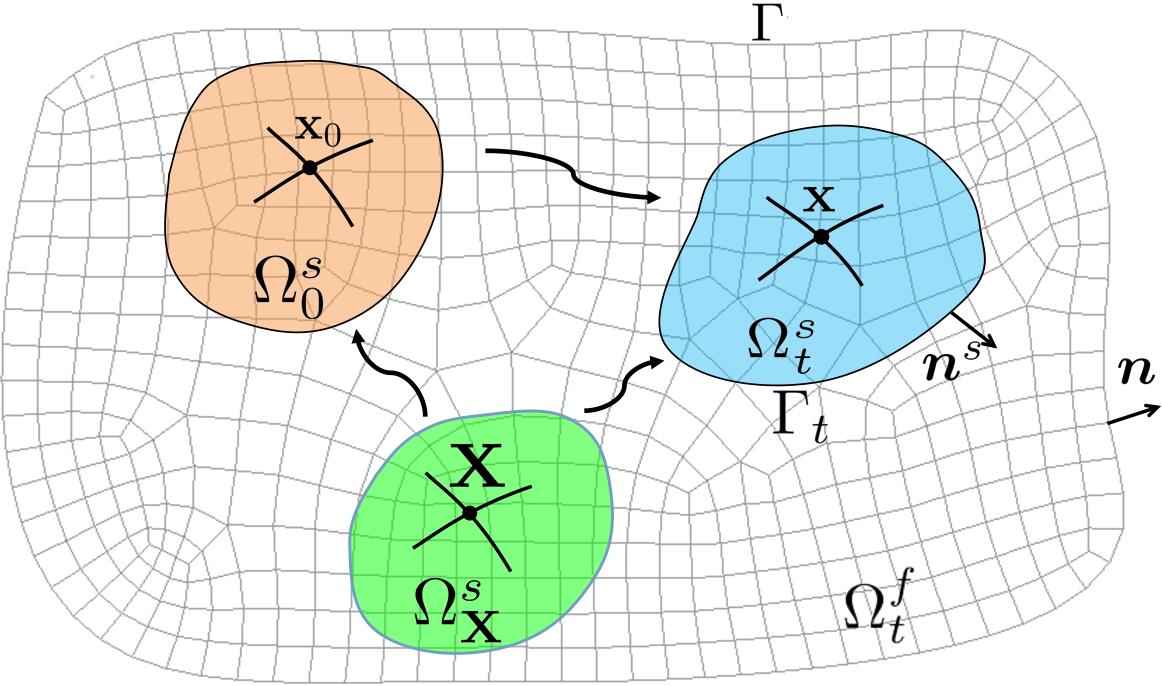}
	\caption {\scriptsize Schematic diagram of FSI, $\Omega=\Omega_t^f\cup \Omega_t^s$.} 
	\label{fig1:Schematic diagram of FSI}
\end{figure}

Let $\rho, \mu, {\bf u}, p, {\bm{\sigma}}$ denote the density, viscosity, velocity, pressure and stress tensor respectively. We assume both an incompressible fluid and incompressible solid, then the conservation of momentum and conservation of mass take the same form as follows:

Momentum equation:
\begin{equation} \label{momentum_equation}
\rho\frac{d{\bf u}}{dt}
=\nabla \cdot {\bm\sigma},
\end{equation}
Continuity equation:
\begin{equation} \label{continuity_equation}
\nabla \cdot {\bf u}=0.
\end{equation}

An incompressible Newtonian constitutive equation in $\Omega^f$ can be expressed as:
\begin{equation} \label{constitutive_fluid}
{\bm\sigma}={\bm\sigma}^f=\mu^f {\rm D}{\bf u}^f-p^f{\bf I},
\end{equation}
where  ${\rm D}{\bf u}=\nabla {\bf u}+\nabla^{\scriptsize T} {\bf u}$. We shall use an incompressible neo-Hookean solid in $\Omega_t^s$ \cite{Boffi_2016, Hecht_2017}, and in common with previous work \cite{ Wang_2013, zhang2004immersed} we also assume the solid has the same viscosity as the fluid. The constitutive equation may be expressed as:
\begin{equation} \label{constitutive_solid}
{\bm\sigma}={\bm\sigma}^s=c_1J^{-1}\left({\bf F}{\bf F}^T-{\bf I}\right)+\mu^f{\rm D}{\bf u}^s-p^s{\bf I},
\end{equation}
where ${\bf F}=\frac{\partial {\bf x}}{\partial {\bf X}}=\frac{\partial {\bf x}}{\partial {\bf x}_0} \frac{\partial {\bf x}_0}{\partial {\bf X}}$=$\nabla_0{\bf x}\nabla_{\bf X}{\bf x}_0$ is the deformation tensor of the solid, and $J=det{\bf F}$ is the determinant of ${\bf F}$. Finally the system is complemented with the following boundary and initial conditions.
\begin{equation}\label{interfaceBC1}
{\bf u}^f={\bf u}^s\quad on \quad  \Gamma_t,
\end{equation}
\begin{equation}\label{interfaceBC2}
{\bf n}^s{\bm \sigma}^f= {\bf n}^s{\bm \sigma}^s\quad on \quad  \Gamma_t,
\end{equation}
\begin{equation}\label{homogeneous_boundary}
{\bf u}^f={\bf 0}\quad on \quad  \Gamma,
\end{equation}
\begin{equation} \label{initialcd_fluid}
\left. {\bf u}^f\right|_{t=0}={\bf u}_0^f,
\end{equation}
\begin{equation} \label{initialcd_solid}
\left. {\bf u}^s\right|_{t=0}={\bf u}_0^s.
\end{equation}
Other boundary conditions are possible on $\Gamma$ but (\ref{homogeneous_boundary}) are used here for simplicity.

\begin{remark}
The corresponding energy function for the hyperelastic stress in (\ref{constitutive_solid}) is defined by \cite{Hesch_2014}:
\begin{equation}
\Psi\left({\bf F}\right)=\frac{c_1}{2}\left(tr_{{\bf F}{\bf F}^T}-d\right)-c_1 ln(J).
\end{equation}
\end{remark}
%%%%%%%%%%%%%%%%%%%%%%%%%%%%%%%%%%%%%%%%%%%%%%%%%%%%%%%%%
\section{Weak formulation}
\label{sec:wfotcl}
The finite element weak form discussed in this section is almost the same as that in \cite{Wang_2017}, the only difference is that we integrate the solid stress in the reference domain, because we shall update the deformation tensor (${\bf F}$-scheme) rather than the solid stress as done in \cite{Wang_2017} (${\bm{\sigma}}$-scheme).
In the following context, let $L^2(\omega)$ be the square integrable functions in domain $\omega$, endowed with norm $\left\|u\right\|_{0,\omega}^2=\int_\omega \left|u\right|^2$ ($u\in L^2(\omega)$). Let $H^1(\omega)=\left\{u:u, \nabla u\in L^2(\omega)\right\}$ with the norm denoted by $\left\|u\right\|_{1,\omega}^2=\left\|u\right\|_{0,\omega}^2+\left\|\nabla u\right\|_{0,\omega}^2$. We also denote by $H_0^1(\omega)$ the subspace of $H^1(\omega)$ whose functions have zero values on the boundary of $\omega$, and denote by $L_0^2(\omega)$ the subspace of $L^2(\omega)$ whose functions have zero mean value.

Let
$
{p}=\left \{ 
\begin{matrix}
{{p}^f \quad in \quad \Omega_t^f} \\
{{p}^s \quad in \quad \Omega_t^s} \\
\end{matrix}\right.
$.
Given ${\bf v}\in H_0^1(\Omega)^d$, we perform the following symbolic operations:
$$
\int_{\Omega_t^f}{\rm Eq.}(\ref{momentum_equation})\cdot{\bf v}d{\bf x}
+\int_{\Omega_t^s}{\rm Eq.}(\ref{momentum_equation})\cdot{\bf v}d{\bf x}.
$$

Integrating the stress terms by parts, the above operations, using constitutive equation (\ref{constitutive_fluid}) and (\ref{constitutive_solid}) and boundary condition (\ref{interfaceBC2}), gives:
\begin{equation}\label{operation1}
\begin{split}
&\rho^f\int_{\Omega}\frac{d{\bf u}}{dt} \cdot{\bf v}d{\bf x}
+\frac{\mu^f}{2}\int_{\Omega}{\rm D}{\bf u}:{\rm D}{\bf v}d{\bf x}
-\int_{\Omega}p\nabla \cdot {\bf v}d{\bf x} \\
&+\rho^{\delta}\int_{\Omega_t^s}\frac{d{\bf u}}{dt}\cdot{\bf v}d{\bf x}
+c_1\int_{\Omega_t^s}J^{-1}\left({\bf F}{\bf F}^T-{\bf I}\right):\nabla{\bf v}d{\bf x}
=0,
\end{split}
\end{equation}
where $\rho^{\delta}=\rho^s-\rho^f$. Note that the integrals on the interface $\Gamma_t$ are cancelled out using boundary condition (\ref{interfaceBC2}). This is not surprising because they are internal forces for the whole FSI system considered here. 

Transforming the integral of the last two terms of (\ref{operation1}) to the reference coordinate system, combined with the following symbolic operations for $q\in L^2(\Omega)$,
$$
-\int_{\Omega_t^f}{\rm Eq.}(\ref{continuity_equation})qd{\bf x}
-\int_{\Omega_t^s}{\rm Eq.}(\ref{continuity_equation})qd{\bf x},
$$ 
leads to the weak form of the FSI system as follows.

\begin{problem}\label{problem_weak_continuous} 
Given ${\bf u}_0$ and $\Omega_0^s$, find ${\bf u}(t)\in H_0^1(\Omega)^d$, $p(t) \in L_0^2(\Omega)$ and $\Omega_t^s$, such that for $\forall {\bf v}\in H_0^1(\Omega)^d$, $\forall q \in L^2(\Omega)$, the following two equations hold:
\begin{equation}\label{weak_form1}
\begin{split}
&\rho^f\int_{\Omega}\frac{\partial{\bf u}}{\partial t} \cdot{\bf v}d{\bf x}
+\rho^f \int_{\Omega}\left({\bf u}\cdot\nabla\right){\bf u}\cdot{\bf v}d{\bf x}
+\frac{\mu^f}{2}\int_{\Omega}{\rm D}{\bf u}:{\rm D}{\bf v}d{\bf x}
-\int_{\Omega}p\nabla \cdot {\bf v}d{\bf x} \\
&+\rho^{\delta}\int_{\Omega_{\bf X}^s}\frac{\partial{\bf u}}{\partial t}\cdot{\bf v}d{\bf X} 
+c_1\int_{\Omega_{\bf X}^s}{\bf F}:\nabla_{\bf X}{\bf v}d{\bf X}
-c_1\int_{\Omega_t^s}J^{-1}\nabla\cdot{\bf v}d{\bf x}
=0,
\end{split}
\end{equation}
 and
\begin{equation}\label{weak_form2}
-\int_{\Omega} q\nabla \cdot {\bf u}d{\bf x}=0.
\end{equation}		
\end{problem}

\begin{remark}
Because domain $\Omega$ is stationary (the Eulerian description will be used) and $\Omega_t^s$ is transient which will be updated by its own velocity (the updated Lagrangian description), there is a convection term from the total derivative of time in $\Omega$, but there is no convection term in $\Omega_t^s$.
\end{remark}

\begin{remark}
Problem \ref{problem_weak_continuous} is equivalent to the equation (12) in \cite{Wang_2017}.
\end{remark}

%%%%%%%%%%%%%%%%%%%%%%%%%%%%%%%%%%%%%%%%%%%%%%%%%%%%%%%%
\section{Discretization in time}
\label{sec_weak_form_time_discretization}
We may use the backward Euler method to discretize Problem \ref{problem_weak_continuous}, and update coordinates of the solid by ${\bf x}_{n+1}={\bf x}_n+\Delta t{\bf u}_{n+1}$. As a result, ${\bf F}$ is updated by ${\bf F}_{n+1}={\bf F}_n+\Delta t\nabla_{\bf X}{\bf u}_{n+1}$, and so,
\begin{equation}\label{fn1}
\int_{\Omega_{\bf X}^s}{\bf F}_{n+1}:\nabla_{\bf X}{\bf v}
=\int_{\Omega_{\bf X}^s}{\bf F}_n:\nabla_{\bf X}{\bf v}
+\Delta t\int_{\Omega_{\bf X}^s}\nabla_{\bf X}{\bf u}_{n+1}:\nabla_{\bf X}{\bf v}.
\end{equation}

Using equation (\ref{fn1}), the discretized weak form corresponding to Problem \ref{problem_weak_continuous} may be expressed as:
\begin{problem}\label{problem_weak_after_time_discretization}
Given ${\bf u}_n$, $p_n$ and $\Omega_n^s$, find ${\bf u}_{n+1}\in H_0^1(\Omega)^d$, $p_{n+1} \in L_0^2(\Omega)$ and $\Omega_{n+1}^s$, such that for $\forall{\bf v}\in H_0^1\left(\Omega\right)^d$, $\forall q\in L^2(\Omega)$, the following four relations hold:
\begin{equation}\label{weak_form1_discretization_backward_Euler}
\begin{split}
&\rho^f\int_{\Omega}\frac{{\bf u}_{n+1}-{\bf u}_n}{\Delta t} \cdot{\bf v}d{\bf x}
+\rho^f \int_{\Omega}\left({\bf u}_{n+1}\cdot\nabla\right){\bf u}_{n+1}\cdot{\bf v}d{\bf x} \\
&+\frac{\mu^f}{2}\int_{\Omega}{\rm D}{\bf u}_{n+1}:{\rm D}{\bf v}d{\bf x}
-\int_{\Omega}p_{n+1}\nabla \cdot {\bf v}d{\bf x}\\
&+\rho^{\delta}\int_{\Omega_{\bf X}^s}\frac{{\bf u}_{n+1}-{\bf u}_n}{\Delta t} \cdot{\bf v}d{\bf X}
+c_1\Delta t\int_{\Omega_{\bf X}^s}\nabla_{\bf X}{\bf u}_{n+1}:\nabla_{\bf X}{\bf v}d{\bf X}\\
&-c_1\int_{\Omega_{n+1}^s}J_{n+1}^{-1}\nabla\cdot{\bf v}d{\bf x}
=-c_1\int_{\Omega_{\bf X}^s}{\bf F}_n\nabla_{\bf X}{\bf v}d{\bf X},
\end{split}
\end{equation}
\begin{equation}\label{weak_form2_discretization_backward_Euler}
-\int_{\Omega} q\nabla \cdot {\bf u}_{n+1}d{\bf x}=0,
\end{equation}
\begin{equation}\label{update_of_solid_mesh_backward_Euler}
\Omega_{n+1}^s=\left\{{\bf x}:{\bf x}={\bf x}_n+\Delta t{\bf u}_{n+1}, {\bf x}_n\in\Omega_n^s \right\},
\end{equation}
and
\begin{equation}
{\bf F}_{n+1}={\bf F}_n+\Delta t\nabla_{\bf X}{\bf u}_{n+1}.
\end{equation}
\end{problem}

\begin{remark}
We shall use a fixed-point iteration at each time step to construct $\Omega_{n+1}^s$ implicitly.
\end{remark}

\begin{remark}
Problem \ref{problem_weak_after_time_discretization} is similar to equation (30) in \cite{Wang_2017}, however here the discretized weak form is expressed as an implicit scheme, and the solid deformation tensor ${\bf F}$ is updated rather than the solid stress ${\bm\sigma}^s$ in \cite{Wang_2017}.
\end{remark}

%%%%%%%%%%%%%%%%%%%%%%%%%%%
\section{Stability by energy estimate}
\label{stability_analysis}
\subsection{Energy conservation in the continuous case}
\label{sec:energy_continuous}
In this section we shall prove that the weak forms (\ref{weak_form1}) and (\ref{weak_form2}), associated with Problem \ref{problem_weak_continuous}, preserve energy.
\label{sec:ecotcl}
\begin{lemma}\label{lemma_potential_energy}
The energy function $\Psi\left({\bf F}\right)$ for the hyperelastic stress satisfies:
\begin{equation}\label{potential_energy_relation}
c_1\int_0^t\int_{\Omega_{\bf X}^s}{\bf F}:{\nabla}_{\bf X}{\bf u}d{\bf X}
-c_1\int_0^t\int_{\Omega_t^s}J^{-1}{\nabla}\cdot{\bf u}d{\bf x}
=\int_{\Omega_{\bf X}^s}\Psi(\bf F)d{\bf X}.
\end{equation}
\end{lemma}
\begin{proof}
Since $\frac{\partial tr_{{\bf F}{\bf F}^T}}{\partial{\bf F}}=2{\bf F}$ and $\frac{\partial \left(det_{\bf F}\right)}{\partial{\bf F}}=det_{\bf F}{\bf F}^{-T}$. Using the fact that ${\bf A:B}=tr_{{\bf AB}^T}$ (${\bf A}$ and ${\bf B}$ are arbitrary matrices), we have:
\begin{equation*}
\begin{split}
&\frac{d}{dt}\int_{\Omega_{\bf X}^s}\Psi({\bf F})d{\bf X}
=\int_{\Omega_{\bf X}^s}\frac{\partial\Psi}{\partial{\bf F}}:\frac{d{\bf F}}{dt}d{\bf X}\\
= & c_1\int_{\Omega_{\bf X}^s}\left({\bf F}-{\bf F}^{-T}\right):\frac{d}{dt}\left({\bf I}+\nabla_{\bf X}{\bf d}\right)d{\bf X}\\
=& c_1\int_{\Omega_{\bf X}^s}{\bf F}:\nabla_{\bf X}{\bf u}d{\bf X}
-c_1\int_{\Omega_t^s}J^{-1}\nabla\cdot{\bf u}d{\bf x},
\end{split}
\end{equation*}
where ${\bf d}$ is displacement of the solid at time $t$.
\end{proof}

\begin{lemma}\label{lemma_convection_zero}
	If $\left({\bf u}, p\right)$ is the solution pair of Problem \ref{problem_weak_continuous}, then
	\begin{equation}\label{convection_zero}
		\int_{\Omega}\left({\bf u}\cdot\nabla\right){\bf u}\cdot{\bf u}d{\bf x}=0.
	\end{equation}
\end{lemma}
\begin{proof}
	First,
	\begin{equation}\label{lemma_convection_zero_proof1}
		\int_{\Omega}\left({\bf u}\cdot\nabla\right){\bf u}\cdot{\bf u}d{\bf x}
		=\int_{\Omega}\nabla\left({\bf u}\otimes{\bf u}\right)\cdot{\bf u}d{\bf x}
		-\int_{\Omega}\left|{\bf u}\right|^2\nabla\cdot{\bf u}d{\bf x}.
	\end{equation}
	Integrate by parts:
	\begin{equation}\label{lemma_convection_zero_proof2}
		\int_{\Omega}\nabla\left({\bf u}\otimes{\bf u}\right)\cdot{\bf u}d{\bf x}
		=\int_{\Gamma}\left|{\bf u}\right|^2{\bf u}\cdot {\bf n}d{\Gamma}
		-\int_{\Omega}\left({\bf u}\cdot\nabla\right){\bf u}\cdot{\bf u}d{\bf x}.
	\end{equation}
According to a Sobolev imbedding theorem \cite[Theorem 6 in Chapter 5] {mitrovic1997fundamentals} and the inclusion between $L^p$ spaces ($L^q \subset L^p$ if $p<q$), we know $H^1(\Omega)\subset L^4(\Omega)$ (for both 2D and 3D). Therefore ${\bf u}\in L^4(\Omega)$, i.e., $\int_\Omega|{\bf u}|^4d{\bf x}<\infty$. That is to say $|{\bf u}|^2\in L^2(\Omega)$. Then we have $\int_{\Omega}\left|{\bf u}\right|^2\nabla\cdot{\bf u}=0$ from (\ref{weak_form2}). We also have $
\int_{\Gamma}\left|{\bf u}\right|^2{\bf u}\cdot {\bf n}=0$ from the boundary condition (\ref{homogeneous_boundary}). Substituting these two equations into (\ref{lemma_convection_zero_proof1}) and (\ref{lemma_convection_zero_proof2}) gives equation (\ref{convection_zero}).
\end{proof}

\begin{proposition} [Energy Conservation] Let $\left({\bf u}, p\right)$ be the solution pair of Problem \ref{problem_weak_continuous}, then
	\begin{equation}\label{energy balance}
		\begin{split}
			&\frac{\rho^f}{2}\int_{\Omega}|{\bf u}|^2d{\bf x}
			+\frac{\mu^f}{2}\int_0^t\int_{\Omega}{\rm D}{\bf u}:{\rm D}{\bf u}d{\bf x} \\
			&+\frac{\rho^\delta}{2}\int_{\Omega_{\bf X}^s}|{\bf u}|^2d{\bf X}
			+\int_{\Omega_{\bf X}^s}\Psi(\bf F)d{\bf X}=0.
		\end{split}
	\end{equation}
\end{proposition}
\begin{proof}
	We first let ${\bf v}={\bf u}$ in (\ref{weak_form1}) and integrate from time $0$ to $t$, then let $q=p$ in (\ref{weak_form2}) and substitute into (\ref{weak_form1}). Finally we can construct the above equation of energy balance due to Lemma \ref{lemma_potential_energy} and \ref{lemma_convection_zero}.
\end{proof}

%-------------------------
\subsection{Stability analysis after time discretization}
\label{sec:energy_time_discretization}
We next demonstrate a similar energy stability result for Problem \ref{problem_weak_after_time_discretization}.
\begin{lemma}\label{energy_estimate}
The trace function $\frac{1}{2}tr\left({{\bf F}{\bf F}^T}\right)$ satisfies:
\begin{equation}
\frac{1}{2}tr\left({{\bf F}_{n+1}{\bf F}_{n+1}^T}\right)-\frac{1}{2}tr\left({{\bf F}_n{\bf F}_n^T}\right)
=
\Delta t{\bf F}_{n+1}:\nabla_{\bf X}{{\bf u}_{n+1}}
-\frac{\Delta t^2}{2}\left|\nabla_{\bf X}{\bf u}_{n+1}\right|^2,
\end{equation}
where $\left|{\bf A}\right|^2=\sum_{ij}a_{ij}^2$ for an arbitrary matrix ${\bf A}=\left[a_{ij}\right]$.
\end{lemma}
\begin{proof}
\begin{equation*}
\begin{split}
&{\bf F}_{n+1}{\bf F}_{n+1}^T-{\bf F}_n{\bf F}_n^T \\
=&{\bf F}_{n+1}{\bf F}_{n+1}^T-\left({\bf F}_{n+1}-\Delta t\nabla_{\bf X}{\bf u}_{n+1}\right)\left({\bf F}_{n+1}-\Delta t\nabla_{\bf X}{\bf u}_{n+1}\right)^T \\
=&\Delta t{\bf F}_{n+1}\nabla_{\bf X}^T{\bf u}_{n+1}+\Delta t\nabla_{\bf X}{\bf u}_{n+1}{\bf F}_{n+1}^T-\Delta t^2\nabla_{\bf X}{\bf u}_{n+1}\nabla_{\bf X}^T{\bf u}_{n+1}.
\end{split}
\end{equation*}
Lemma \ref{energy_estimate} holds due to
\begin{equation*}
\frac{1}{2}tr{\left({{\bf F}_{n+1}{\bf F}_{n+1}^T}-{{\bf F}_n{\bf F}_n^T}\right)}
=\Delta t\cdot tr\left({{\bf F}_{n+1}\nabla_{\bf X}^T{\bf u}_{n+1}}\right)
-\frac{\Delta t^2}{2}\left|\nabla_{\bf X}{\bf u}_{n+1}\right|^2.
\end{equation*}
\end{proof}

%%%-----------------------
\begin{lemma}
The log-determinant function $ln\left(det{\bf F}\right)$ satisfies:
\begin{equation*}
ln(det{{\bf F}_{n+1}})-ln(det{{\bf F}_n})
\ge 
\Delta t\nabla\cdot{\bf u}_{n+1}
-\frac{\Delta t^2}{2}\left|{\bf F}_{n+1}^{-1}\nabla_{\bf X}{\bf u}_{n+1}\right|^2.
\end{equation*}
\begin{proof}
Use the fact that function $ln(det{\bf Y})$ is concave over the set of positive definite matrices \cite [Chapter 3] {Boyd}. Let ${\bf B}={\bf F}{\bf F}^T$, 
$\mathcal{F}({\bf B})=\frac{1}{2}ln\left(det{\bf B}\right)=ln\left(det{\bf F}\right)$ and $w(\xi)=\mathcal{F}\left({\bf B}_n+\xi\left({\bf B}_{n+1}-{\bf B}_n\right)\right)$, then
\begin{equation*}
\begin{split}
w^\prime(\xi)=\frac{d\mathcal{F}}{d{\bf B}}:\left({\bf B}_{n+1}-{\bf B}_n\right)=\frac{1}{2}\left({\bf B}_n+\xi\left({\bf B}_{n+1}-{\bf B}_n\right)\right)^{-1}:\left({\bf B}_{n+1}-{\bf B}_n\right).
\end{split}
\end{equation*}
According to the property of concave functions, we have $w(1)-w(0) \ge w^\prime(1)$, this is to say:
\begin{equation*}
\begin{split}
& ln(det{{\bf F}_{n+1}})-ln(det{{\bf F}_n})
=\mathcal{F}\left({\bf B}_{n+1}\right)-\mathcal{F}\left({\bf B}_{n}\right)\\
\ge&\frac{1}{2}{\bf B}_{n+1}^{-1}:\left({\bf B}_{n+1}-{\bf B}_n\right)
=\frac{1}{2}tr\left({\bf I}-{\bf B}_{n+1}^{-1}{\bf B}_n\right) \\
=&\frac{1}{2}tr\left({\bf I}-{\bf B}_{n+1}^{-1}\left({\bf F}_{n+1}-\Delta t\nabla_{\bf X}{\bf u}_{n+1}\right)\left({\bf F}_{n+1}^T-\Delta t\nabla_{\bf X}^T{\bf u}_{n+1}\right)\right) \\
=&\frac{\Delta t}{2}tr\left(  {\bf F}_{n+1}^{-T}\nabla_{\bf X}^T{\bf u}_{n+1}  +{\bf F}_{n+1}^{-T}{\bf F}_{n+1}^{-1}\nabla_{\bf X}{\bf u}_{n+1}{\bf F}_{n+1}^T \right) \\
-&\frac{\Delta t^2}{2}tr\left( {\bf F}_{n+1}^{-T}{\bf F}_{n+1}^{-1}\nabla_{\bf X}{\bf u}_{n+1}\nabla_{\bf X}^T{\bf u}_{n+1}  \right) \\
=&\Delta t\nabla\cdot{\bf u}_{n+1} 
-\frac{\Delta t^2}{2}\left|{\bf F}_{n+1}^{-1}\nabla_{\bf X}{\bf u}_{n+1}\right|^2.
\end{split}
\end{equation*}
In the above, we use the trace property of cyclic permutations: $tr\left({\bf A_1}{\bf A_2}{\bf A_3}\right)=tr\left({\bf A_2}{\bf A_3}{\bf A_1}\right)=tr\left({\bf A_3}{\bf A_1}{\bf A_2}\right)$.
\end{proof}
\end{lemma}

From the above two lemmas, we have:
\begin{proposition}\label{estimate_phi}
The energy function $\Psi\left({\bf F}\right)$ for the hyperelastic stress satisfies:
\begin{equation}
\begin{split}
&\int_{\Omega_{\bf X}^s}\Psi\left({\bf F}_{n+1}\right)d{\bf X}-\int_{\Omega_{\bf X}^s}\Psi\left({\bf F}_n\right)d{\bf X}\\
&\le \Delta tc_1\int_{\Omega_{\bf X}^s}{\bf F}_{n+1}:\nabla_{\bf X}{\bf u}_{n+1}d{\bf X}
-\Delta tc_1\int_{\Omega_{n+1}^s}J_{n+1}^{-1}\nabla\cdot{\bf u}_{n+1}d{\bf x}+R_{n+1},
\end{split}
\end{equation}
where 
\begin{equation}\label{residual}
R_{n+1}=\frac{c_1\Delta t^2}{2}\int_{\Omega_{\bf X}^s}\left(\left|{\bf F}_{n+1}^{-1}\nabla_{\bf X}{\bf u}_{n+1}\right|^2
-\left|\nabla_{\bf X}{\bf u}_{n+1}\right|^2\right)d{\bf X}.
\end{equation}
\end{proposition}

Similarly to Lemma \ref{lemma_convection_zero}, we have:
\begin{lemma}\label{lemma_convection_zero_discretization_in_time}
	If $\left({\bf u}_{n+1}, p_{n+1}\right)$ is the solution pair of Problem \ref{problem_weak_after_time_discretization}, then
	\begin{equation}\label{convection_zero_discretization_in_time_euler}
	\int_{\Omega}\left({\bf u}_{n+1}\cdot\nabla\right){\bf u}_{n+1}\cdot{\bf u}_{n+1}d{\bf x}=0.
	\end{equation}
\end{lemma}

\begin{proposition} [Energy Nonincreasing] \label{lec_backward_Euler}
	Let $\left({\bf u}_{n+1}, p_{n+1}\right)$ be the solution pair of Problem \ref{problem_weak_after_time_discretization}. If $\rho^\delta\ge 0$, then
	\begin{equation}\label{energy_estimate_after_time_discretization_backward_Euler}
	\begin{split}
	&\frac{\rho^f}{2}\int_{\Omega}\left|{\bf u}_{n+1}\right|^2d{\bf x}
	+\frac{\rho^\delta}{2}\int_{\Omega_{\bf X}^s}\left|{\bf u}_{n+1}\right|^2d{\bf X}
	+\int_{\Omega_{\bf X}^s}\Psi\left({\bf F}_{n+1}\right)d{\bf X}  \\
	&+\frac{\Delta t\mu^f}{2}\sum_{k=1}^{n+1}\int_{\Omega}{\rm D}{\bf u}_k:{\rm D}{\bf u}_kd{\bf x}\\
	&\le \frac{\rho^f}{2}\int_{\Omega}\left|{\bf u}_n\right|^2d{\bf x}
	+\frac{\rho^\delta}{2}\int_{\Omega_{\bf X}^s}\left|{\bf u}_n\right|^2d{\bf X}
	+\int_{\Omega_{\bf X}^s}\Psi\left({\bf F}_n\right)d{\bf X}\\
	&+\frac{\Delta t\mu^f}{2}\sum_{k=1}^n\int_{\Omega}{\rm D}{\bf u}_k:{\rm D}{\bf u}_kd{\bf x}+R_{n+1},
	\end{split}
	\end{equation}
where $R_{n+1}$ is defined in equation (\ref{residual}).
\end{proposition}
\begin{proof}
Let ${\bf v}={\bf u}_{n+1}$ in (\ref{weak_form1_discretization_backward_Euler}) and multiply $\Delta t$ on both side of the equation, and then let $q=p_{n+1}$ in (\ref{weak_form2_discretization_backward_Euler}) and substitute into equation (\ref{weak_form1_discretization_backward_Euler}), we get:	
	\begin{equation}\label{one_equation_backward_Euler}
	\begin{split}
	&\rho^f\int_{\Omega}\left({\bf u}_{n+1}-{\bf u}_n\right)\cdot{\bf u}_{n+1}d{\bf x}
	+\frac{\Delta t\mu^f}{2}\int_{\Omega}{\rm D}{\bf u}_{n+1}:{\rm D}{\bf u}_{n+1}d{\bf x}\\
	&+\rho^{\delta}\int_{\Omega_{\bf X}^s}\left({\bf u}_{n+1}-{\bf u}_n\right)\cdot{\bf u}_{n+1}d{\bf X}\\
    &+c_1\Delta t\int_{\Omega_{\bf X}^s}{\bf F}_{n+1}:\nabla_{\bf X}{\bf u}_{n+1}d{\bf X}
    -c_1\Delta t\int_{\Omega_{n+1}^s}\nabla\cdot{\bf u}_{n+1}d{\bf x}
    =0.
	\end{split}
	\end{equation}
	Using the Cauchy-Schwarz inequality and the fact $ab\le\frac{a^2+b^2}{2}$, we have:
	\begin{equation*}
	\int_{\omega}{\bf u}_n\cdot{\bf u}_{n+1}d{\bf x}
	\le \left\|{\bf u}_n\right\|_{0,\omega} \left\|{\bf u}_{n+1}\right\|_{0,\omega}
	\le \frac{\|{\bf u}_n\|_{0,\omega}^2+\|{\bf u}_{n+1}\|_{0,\omega}^2}{2},
	\end{equation*}
	where $\omega=\Omega$ or $\Omega_{n+1}^s$. Substituting the above relation into (\ref{one_equation_backward_Euler}), we get (\ref{energy_estimate_after_time_discretization_backward_Euler}) due to Proposition \ref{estimate_phi} and Lemma \ref{lemma_convection_zero_discretization_in_time}.
\end{proof}

\begin{remark}\label{energy_noincreasing_remark}
Relation (\ref{energy_estimate_after_time_discretization_backward_Euler}) does not exactly show energy nonincreasing, because we do not know whether $R_{n+1}$ is greater or less than $0$. However, $R_{n+1}$ is $O\left(\Delta t^2\right)$ and this will be demonstrated in section \ref{sec:numerical_exs} by numerical tests. In order to test the energy property, let us use the following notation for the different contributions to the total energy in (\ref{energy_estimate_after_time_discretization_backward_Euler}): (1) Kinetic energy of fluid plus fictitious fluid 
$
E_k(\Omega)=\frac{\rho^f}{2}\int_{\Omega}\left|{\bf u}_{n}\right|^2d{\bf x};
$
(2) Kinetic energy of solid minus fictitious fluid
$
E_k(\Omega_{\bf X}^s)=\frac{\rho^\delta}{2}\int_{\Omega_{\bf X}^s}\left|{\bf u}_{n}\right|^2d{\bf X};
$
(3) Viscous dissipation
$
E_d(\Omega)=\frac{\Delta t\mu^f}{2}\sum_{k=0}^{n}\int_{\Omega}{\rm D}{\bf u}_{k}:{\rm D}{\bf u}_{k}d{\bf x};
$
(4) Potential energy of the solid
$
E_p(\Omega_{\bf X}^s)=\int_{\Omega_{\bf X}^s}\Psi\left({\bf F}_{n}\right)d{\bf X}.
$
Denote the total energy as
$
E_{total}=E_k(\Omega)+E_k(\Omega_{\bf X}^s)+E_d(\Omega)+E_p(\Omega_{\bf X}^s),
$
and the energy ratio as:
\begin{equation}\label{energy estimate_after_time_discretization_closed}
E_{ratio}=\frac{E_{total}(t_n)}{E_{total}(t_0)}.
\end{equation}
We shall numerically demonstrate that $E_{ratio}$ is nonincreasing in section \ref{sec:numerical_exs}.
\end{remark}

%%%%%%%%%%%%%%%%%%%%%%%%%%%%%%%%%%%%%%%%%%%%%%%%%%%%%%%%
\section{Discretization in space}
\label{sec:weak_form_space_discretization}
We shall use a fixed Eulerian mesh for $\Omega$ and an updated Lagrangian mesh for $\Omega_{n+1}^s$ to discretize Problem \ref{problem_weak_after_time_discretization}. First, we discretize $\Omega$ as $\Omega^h$ with the corresponding finite element spaces as
$$
V^h(\Omega^h)=span\left\{\varphi_1,\cdots,\varphi_{N^u}\right\} \subset H_0^1\left(\Omega\right)
$$
and
$$
L^h(\Omega^h)=span\left\{\phi_1,\cdots,\phi_{N^p}\right\} \subset L_0^2\left(\Omega\right).
$$
The approximated solution ${\bf u}^h$ and $p^h$ can be expressed in terms of these basis functions as
\begin{equation}\label{uh}
{\bf u}^h({\bf x})=\sum_{i=1}^{N^u}{\bf u}({\bf x}_i)\varphi_i({\bf x}), \quad
p^h({\bf x})=\sum_{i=1}^{N^p}p({\bf x}_i)\phi_i({\bf x}).
\end{equation}

We further discretize $\Omega_0^s$ as $\Omega_0^{sh}$ with the corresponding finite element spaces as:
$$
V^{sh}(\Omega_0^{sh})=span\left\{\varphi_1^s,\cdots,\varphi_{N^s}^s\right\} \subset H^1\left(\Omega_0^s\right),
$$
then move the vertices of each element of $\Omega_{n}^{sh}$ by their own velocities to get $\Omega_{n+1}^{sh}$,
and approximate $\left.{\bf u}^h({\bf x})\right|_{{\bf x}\in\Omega_{n+1}^{sh}}$ as:
\begin{equation}\label{ush}
{\bf u}^{sh}\left({\bf x}\right)
=\sum_{i=1}^{N^s}{\bf u}^h({\bf x}_i^s)\varphi_i^s({\bf x})
=\sum_{i=1}^{N^s}\sum_{j=1}^{N^u}{\bf u}({\bf x}_j)\varphi_j({\bf x}_i^s)\varphi_i^s({\bf x}),
\end{equation}
where ${\bf x}_i^s$ is the nodal coordinate of the solid mesh. Notice that the above approximation defines an $L^2$ projection $P_{n+1}$ from $V^h\left(\Omega^h\right)^d$ to $V^{sh}\left(\Omega_{n+1}^{sh}\right)^d$: $P_{n+1}\left({\bf u}^h({\bf x})\right)={\bf u}^{sh}\left({\bf x}\right)$,

We then discretize Problem \ref{problem_weak_after_time_discretization} in space as follows.
\begin{problem}\label{problem_weak_after_space_discretization}
	Given ${\bf u}_n^h$, $p_n^h$ and $\Omega_n^{sh}$, find ${\bf u}_{n+1}^h\in V^h(\Omega^h)^d$, $p_{n+1}^h \in L^h(\Omega^h)$ and $\Omega_{n+1}^{sh}$, such that for $\forall{\bf v}\in V^h(\Omega^h)^d$, $\forall q\in L^h(\Omega^h)$, the following four relations hold:
	\begin{equation}\label{weak_form1_discretization_space}
	\begin{split}
	&\rho^f\int_{\Omega^h}\frac{{\bf u}_{n+1}^h-{\bf u}_n^h}{\Delta t} \cdot{\bf v}d{\bf x}
	+\rho^f \int_{\Omega^h}\left({\bf u}_{n+1}^h\cdot\nabla\right){\bf u}_{n+1}^h\cdot{\bf v}d{\bf x}\\
	&+\frac{\mu^f}{2}\int_{\Omega^h}{\rm D}{\bf u}_{n+1}^h:{\rm D}{\bf v}d{\bf x}
	-\int_{\Omega^h}p_{n+1}^h\nabla \cdot {\bf v}d{\bf x}\\
	&+\rho^{\delta}\int_{\Omega_{\bf X}^{sh}}\frac{{\bf u}_{n+1}^{sh}-{\bf u}_n^{sh}}{\Delta t} \cdot{\bf v}^sd{\bf X}
	+c_1\Delta t\int_{\Omega_{\bf X}^{sh}}\nabla_{\bf X}{\bf u}_{n+1}^{sh}:\nabla_{\bf X}{\bf v}^sd{\bf X}\\
	&-c_1\int_{\Omega_{n+1}^{sh}}J_{n+1}^{-1}\nabla\cdot{\bf v}^sd{\bf x}
	=-c_1\int_{\Omega_{\bf X}^{sh}}{\bf F}_n^{sh}:\nabla_{\bf X}{\bf v}^sd{\bf X},
	\end{split}
	\end{equation}
	\begin{equation}\label{weak_form2_discretization_space}
	-\int_{\Omega} q\nabla \cdot {\bf u}_{n+1}^hd{\bf x}=0,
	\end{equation}	
\begin{equation}\label{updating_disretizatoin_space}
\Omega_{n+1}^{sh}=\left\{{\bf x}:{\bf x}={\bf x}_n+\Delta t{\bf u}_{n+1}^{sh}, {\bf x}_n\in\Omega_n^{sh} \right\},
\end{equation}
and
\begin{equation}\label{update_f}
{\bf F}_{n+1}^{sh}={\bf F}_n^{sh}+\Delta t\nabla_{\bf X}{\bf u}_{n+1}^{sh},
\end{equation}	
where ${\bf u}^{sh}=P_{n+1}\left({\bf u}^h\right)$ and ${\bf v}^s=P_{n+1}\left({\bf v}\right)$.
\end{problem}

\begin{remark}	
The proof of the energy estimate (Proposition \ref{lec_backward_Euler}) for the spatially continuous case can also be applied to the discrete case (see \ref{appendix_space_dis}).
\end{remark}	

\begin{remark}
There are two sources of nonlinearity in Problem \ref{problem_weak_after_space_discretization}: the convection term and the moving solid domain. We can accommodate these by moving the convection term to the right-hand side of the equation, and using a fixed-point iteration to construct $\Omega_{n+1}^s$ in order to solve the nonlinear system at each time step. For other methods to treat convection, readers may refer to \cite{pironneau1989finite, Zienkiewic2014}. We shall only use this fully implicit implementation to consider low Reynolds number ($Re\approx 50$) cases in this paper in order to test the energy stability.
\end{remark}

\begin{remark}
A two-step explicit splitting scheme (${\bf F}$-scheme) is discussed in \ref{appendix_explicit} with corresponding energy analysis. This scheme is similar to that in \cite{Wang_2017} (${\bm\sigma}$-scheme), which may be adapted to problems at large Reynolds number (see \cite{Wang_2017} for more examples).
\end{remark}

%%%%%%%%%%%%%%%%%%%%%%%%%%%%%%%%%%%%%%%%%%
\section{Numerical experiments}
\label{sec:numerical_exs}
In this section, we focus on validation of the energy stability of the proposed numerical method in two and three dimensions. For more two-dimensional numerical examples and validation of the basic algorithm see \cite{Wang_2017}. We shall use linear triangles (2D) or linear tetrahedra (3D) to discretize the solid domain $\Omega_0^s$. In domain $\Omega$, the $P_2/(P_1+P_0)$ elements will be used, i.e., the standard $Taylor$-$Hood$ element $P_2P_1$ is enriched by a constant $P_0$ for approximation of the pressure. This element has the property of local mass conservation and the constant $P_0$ may better capture the element-based jump of pressure \cite{Arnold_2002, Boffi_2011}. We shall demonstrate the improvement of mass conservation and energy conservation by using the $P_2/(P_1+P_0)$ elements compared to the $P_2P_1$ elements. We shall also validate that the total energy is nonincreasing as stated in Proposition \ref{lec_backward_Euler} and Remark \ref{energy_noincreasing_remark}.

\subsection{Oscillating disc driven by an initial kinetic energy (activated disc)}
\label{subsec:oddbaike}
In this test, we consider an enclosed flow (${\bf n}\cdot{\bf u}=0$) in $\Omega=[0,1]\times[0,1]$ with a periodic boundary condition. A solid disc is initially located in the middle of the square $\Omega$ and has a radius of $0.2$. The initial velocity of the fluid and solid are prescribed by the following stream function
\begin{equation*}
\Psi=\Psi_0{\rm sin}(ax){\rm sin}(by),
\end{equation*}
where $\Psi_0=5.0\times10^{-2}$ and $a=b=2\pi$. In this test, $\rho^f=1$, $\mu^f=0.01$, $\rho^s=1.5$ and $c_1=1$. In order to visualize the flow a snapshot of the velocity and deformation fields is presented in Figure \ref{snapshot_of_fluid}, and the evolution of energy is presented in Figure \ref{energy_evolution} using a $50\times 50$ mesh (biquadratic squares for the fluid velocity and $3052$ bilinear triangles for the solid velocity).

\begin{figure}[h!]
	\begin{minipage}[t]{0.5\linewidth}
		\centering  
		\includegraphics[width=1.8in,angle=0]{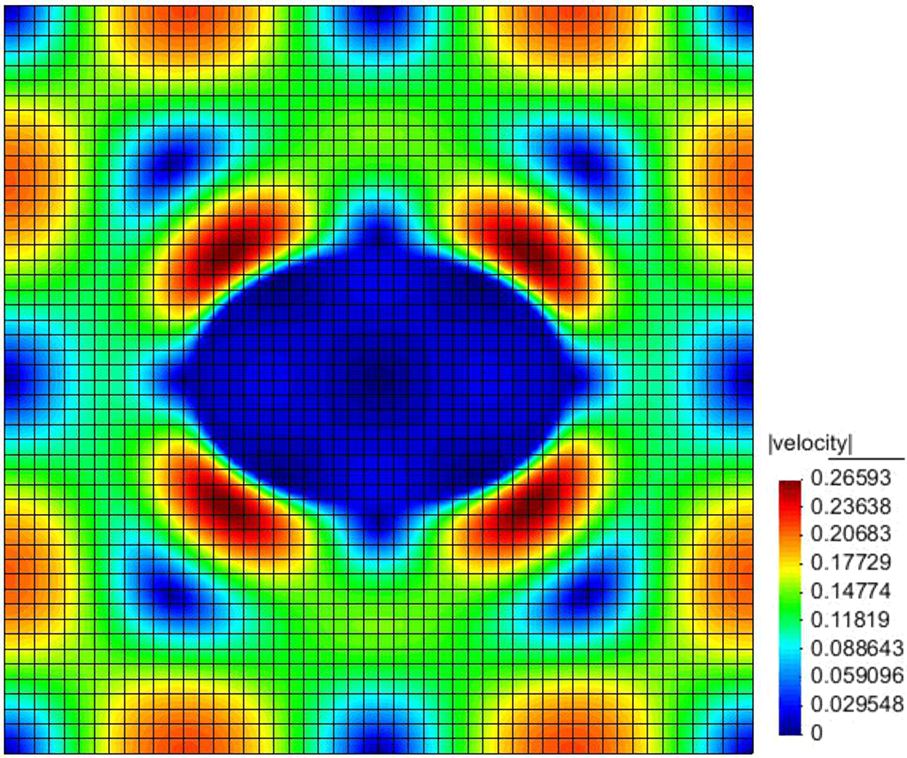}	
		\captionsetup{justification=centering}
		\caption*{\scriptsize(a) Velocity norm on the fluid mesh,}
	\end{minipage}
	\begin{minipage}[t]{0.5\linewidth}
		\centering  
		\includegraphics[width=2.0in,angle=0]{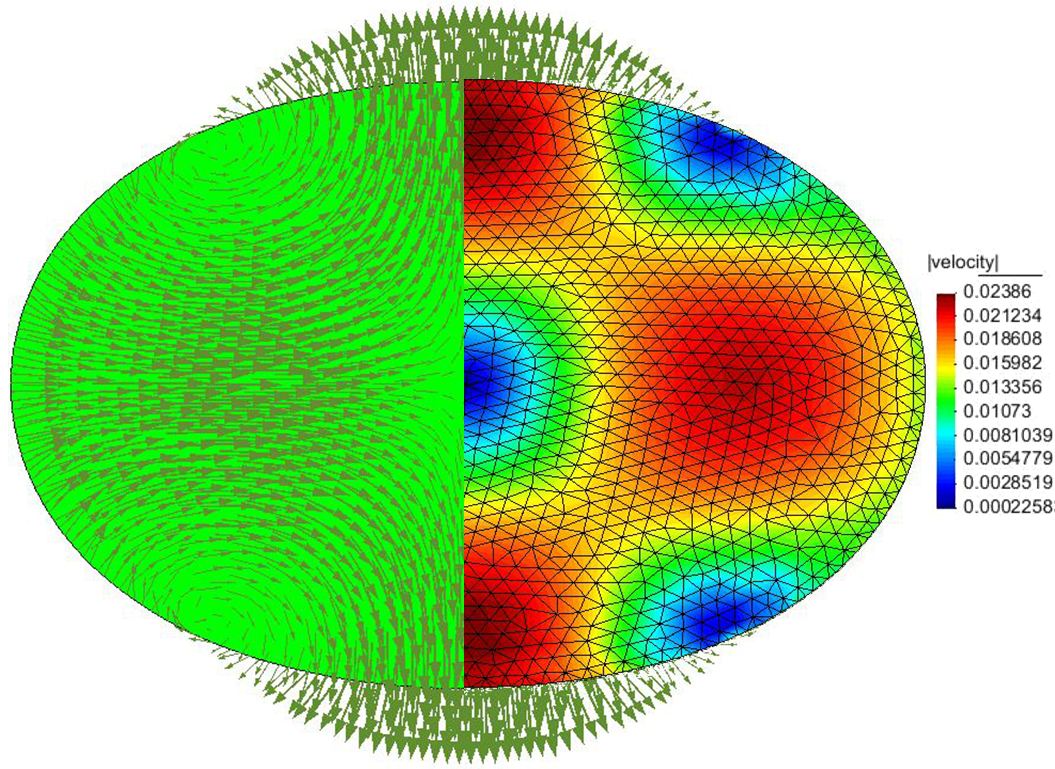}
		\caption*{\scriptsize(b) Distribution of velocity on the solid mesh.}
	\end{minipage}     		
	\captionsetup{justification=centering}
	\caption {\scriptsize Snapshot at $t=0.25$, $\Delta t=5.0\times10^{-3}$.} 
	\label{snapshot_of_fluid}
\end{figure}

\begin{figure}[h!]
\centering  
\includegraphics[width=2.0in,angle=0]{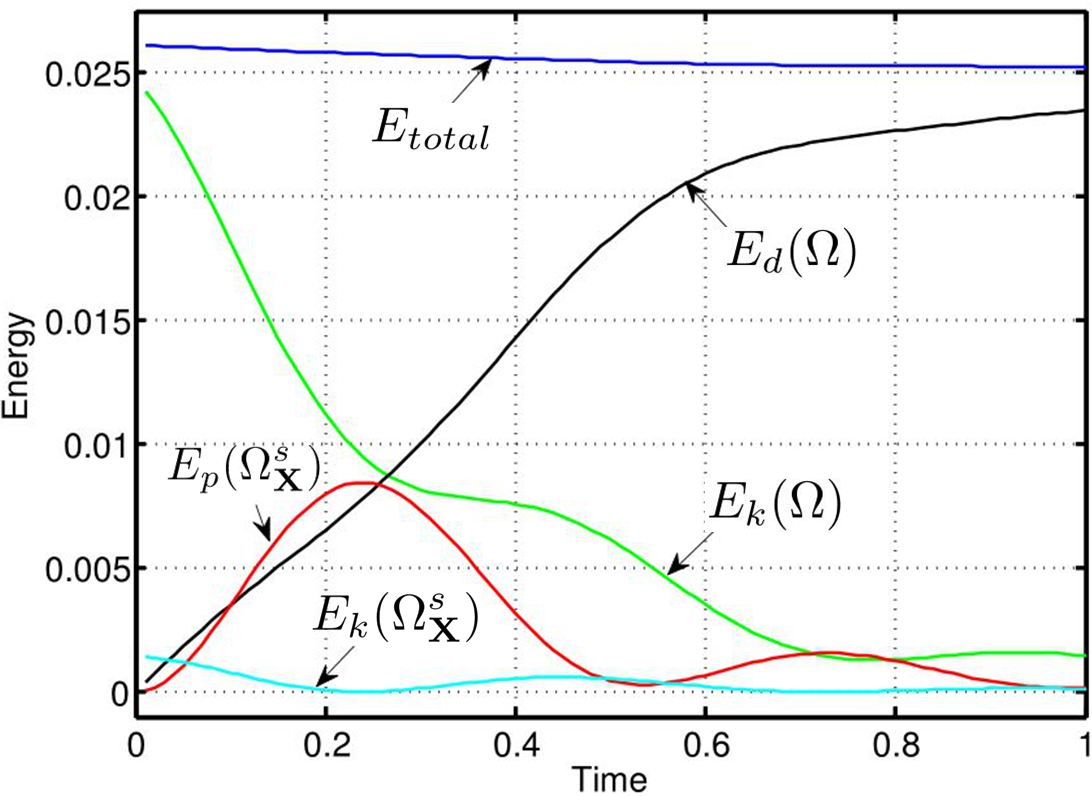}	    		
\captionsetup{justification=centering}
\caption {\scriptsize Evolution of energy, $\Delta t=5.0\times 10^{-3}$.} 
\label{energy_evolution}
\end{figure}

We commence by comparing $P_2/P_1$ elements and $P_2/(P_1+P_0)$ elements. The evolution of mass variation and energy ratio are demonstrated in Figure \ref{mass_energy_conservation}, from which it can be seen that the enrichment of the pressure field by a constant $P_0$ has an effect of stabilizing the mass and energy evolution. In addition, this enrichment of the pressure field dramatically improves the mass conservation, although the effect for energy conservation is not obvious. Then using element $P_2/(P_1+P_0)$, time convergence of the total energy can be observed from Figure \ref{time_mesh_conservation} (a), from which we can see a nonincreasing energy and a first order time convergence for both the implicit and explicit scheme (see \ref{appendix_explicit} for the energy estimate of the explicit scheme). It can be seen from Figure \ref{time_mesh_conservation} (b) that the residual term defined in (\ref{residual}) is very small and converges rapidly to zero when reducing $\Delta t$ ($\sim O(\Delta t^2)$).

\begin{figure}[h!]
	\begin{minipage}[t]{0.5\linewidth}
		\centering  
		\includegraphics[width=2.0in,angle=0]{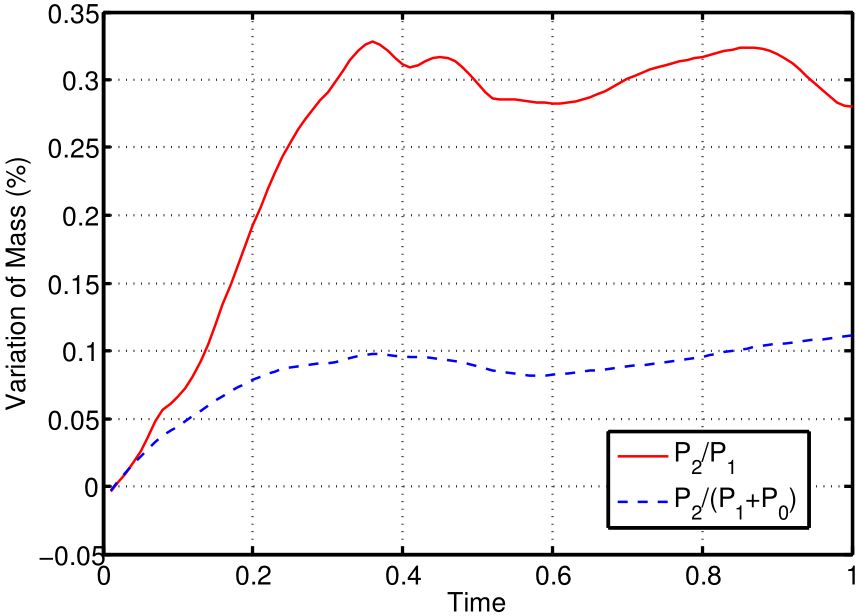}	
		\captionsetup{justification=centering}
		\caption*{\scriptsize(a) Variation of mass against time,}
	\end{minipage}
	\begin{minipage}[t]{0.5\linewidth}
		\centering  
		\includegraphics[width=1.95in,angle=0]{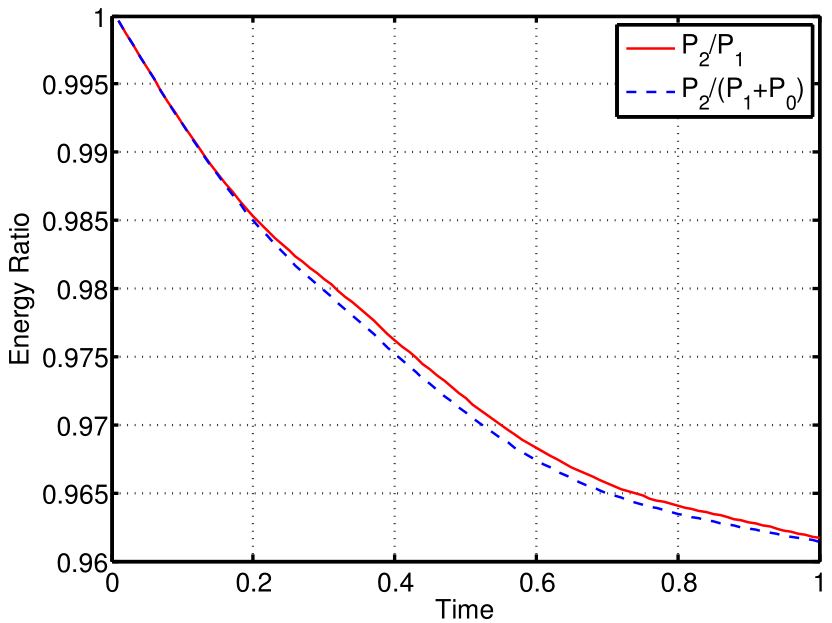}
		\caption*{\scriptsize(b) Energy ratio (see (\ref{energy estimate_after_time_discretization_closed})) against time.}
	\end{minipage}     		
	\captionsetup{justification=centering}
	\caption {\scriptsize Variation of mass and energy, $\Delta t=5.0\times 10^{-3}$.} 
	\label{mass_energy_conservation}
\end{figure}

\begin{figure}[h!]
	\begin{minipage}[t]{0.5\linewidth}
		\centering  
		\includegraphics[width=2.0in,angle=0]{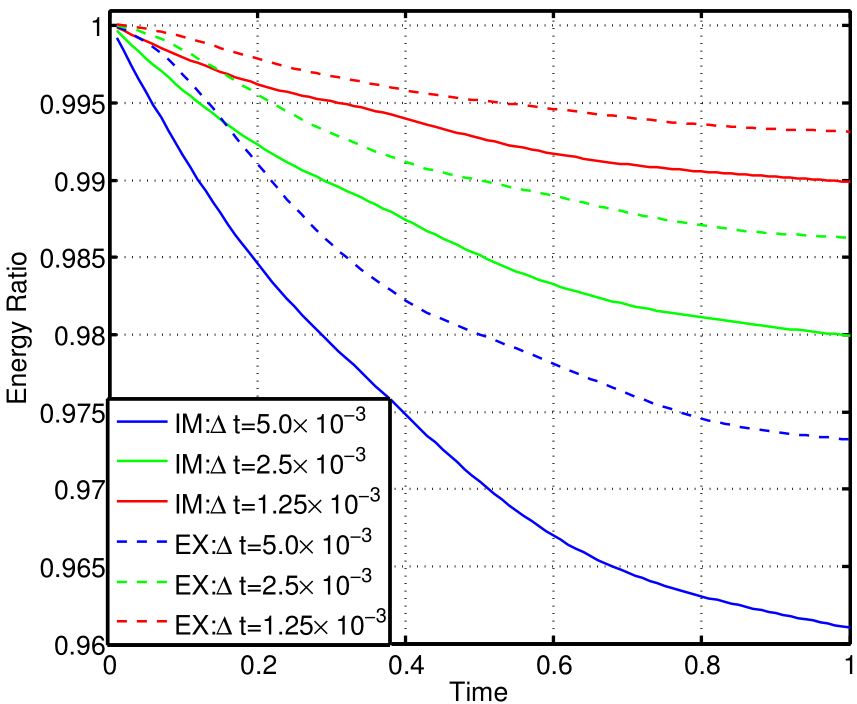}	
		\captionsetup{justification=centering}
		\caption*{\scriptsize(a) Energy ratio against time (defined in (\ref{energy estimate_after_time_discretization_closed})),}
	\end{minipage}
\begin{minipage}[t]{0.5\linewidth}	
	\centering  
	\includegraphics[width=1.9in,angle=0]{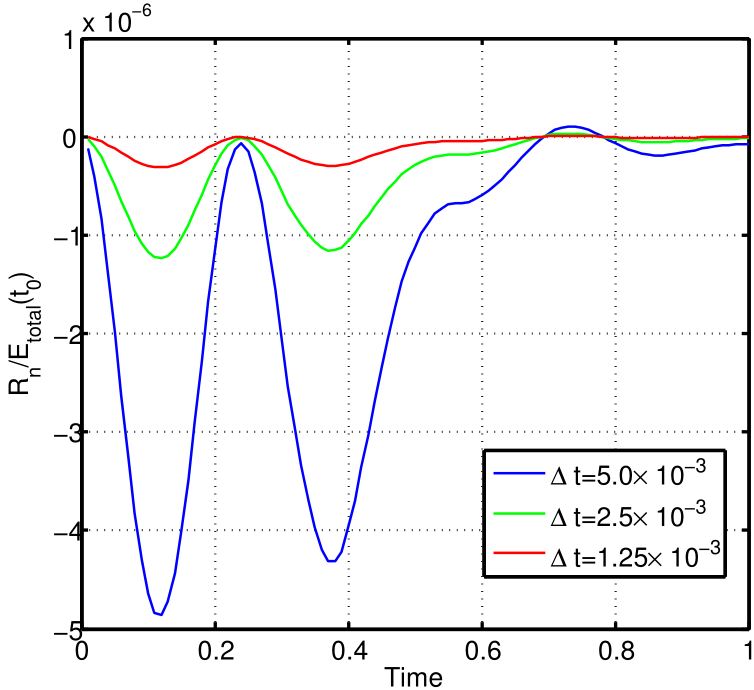}	    		
	\captionsetup{justification=centering}
	\caption*{\scriptsize (b) $R_n/E_{total}\left(t_0\right)$ against time (defined in (\ref{residual})).} 
\end{minipage}    		
	\captionsetup{justification=centering}
	\caption {\scriptsize Evolution of the energy ratio and residual $R_n$ for the test problem of activated disc.} 
	\label{time_mesh_conservation}
\end{figure}

\subsection{Oscillating disc driven by an initial potential energy (stretched disc)}
\label{subsec:odbype}
In the previous example, the disc oscillates because a kinetic energy is prescribed for the FSI system at the beginning. In this test, we shall stretch the disc and create a potential energy in the solid, then release it causing the disc to oscillate due to this potential energy. The computational domain is a square $\Omega=[0,1]\times[0,1]$. One quarter of a solid disc is located in the left-bottom corner of the square, and initially stretched as an ellipse as shown in Figure \ref{figex2}. Notice the equation of an ellipse $\frac{x^2}{a^2}+\frac{y^2}{b^2}=1$ and its area $\pi ab$, hence we ensure that this stretch does not change mass of the solid.
\begin{figure}[h!]
\centering
\includegraphics[width=2.2in,angle=0]{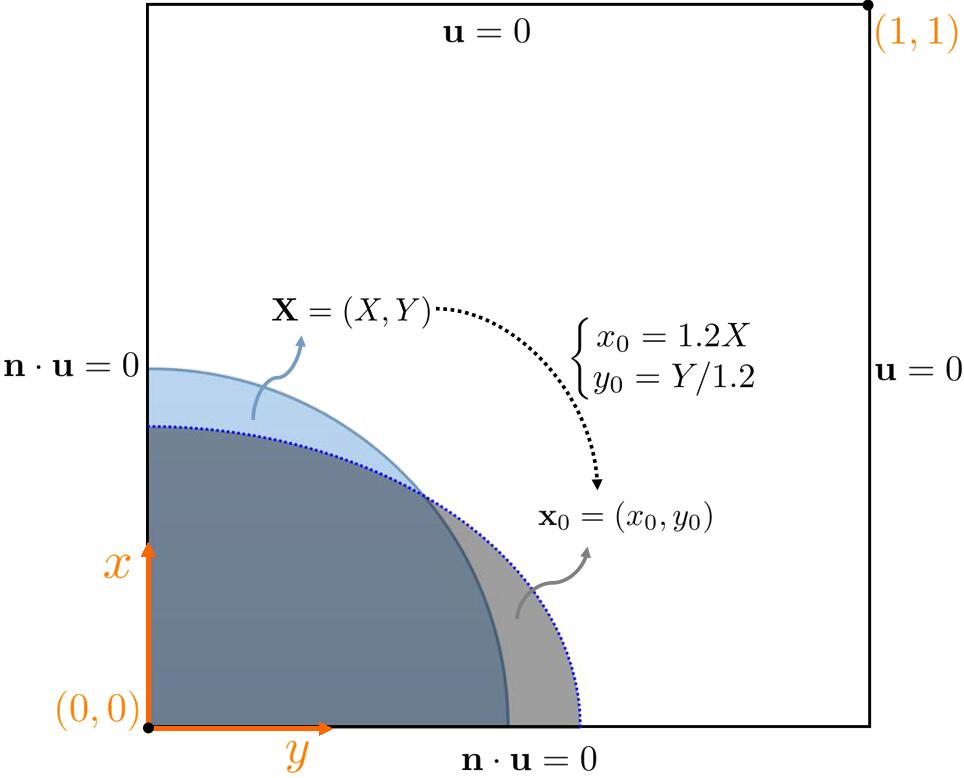}
\caption {\scriptsize Computational domain and boundary conditions for test problem \ref{subsec:odbype} (stretched disc).} 
\label{figex2}
\end{figure}

We choose $\rho^f=1$, $\mu^f=0.01$, $\rho^s=2$ and $c_1=2$. The fluid adopts a mesh of $66\times 66$ biquadratic squares, and the solid has similar node density ($8206$ linear triangles) as the fluid. A snapshot of pressure on the fluid mesh and corresponding solid deformation with its velocity norm are displayed in Figure \ref{Pressure and Velocity}, and the evolution of energy is presented in Figure \ref{energy_evolution_stretch}. The nonincreasing total energy can be observed from Figure \ref{time_mesh_conservation_stretch} (a) for both the implicit and explicit scheme (see \ref{appendix_explicit} for energy estimate of the explicit scheme). It can be seen from Figure \ref{time_mesh_conservation_stretch} (b) that the residual term defined in (\ref{residual}) is very small and converges rapidly to zero when reducing $\Delta t$.

\begin{figure}[h!] 
	\begin{minipage}[t]{0.5\linewidth}
		\centering  
		\includegraphics[width=2.0in,angle=0]{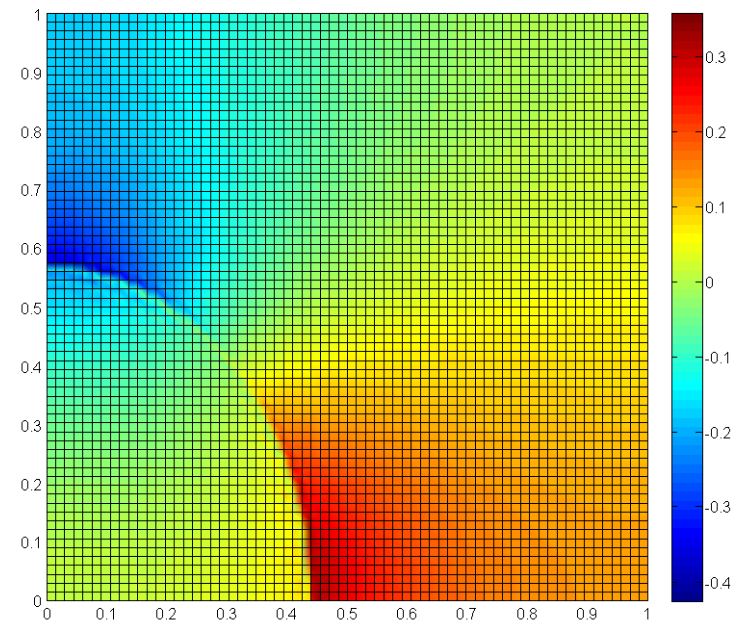}
		\caption*{\scriptsize(a) Distribution of pressures on the fluid mesh,}
	\end{minipage}
	\begin{minipage}[t]{0.5\linewidth}
		\centering  
		\includegraphics[width=1.5in,angle=0]{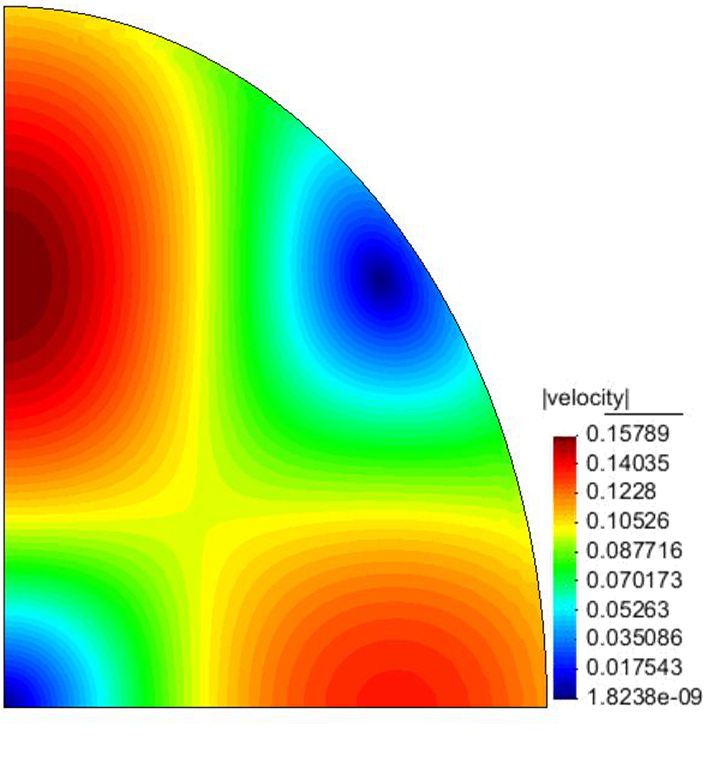}
		\caption*{\scriptsize(b) velocity norm on the solid.}
	\end{minipage}   		
	\captionsetup{justification=centering}
	\caption {\scriptsize A snapshot at $t=1$, $\Delta t=5.0\times 10^{-3}$.} 
	\label{Pressure and Velocity}
\end{figure}

\begin{figure}[h!]
	\centering  
	\includegraphics[width=2.0in,angle=0]{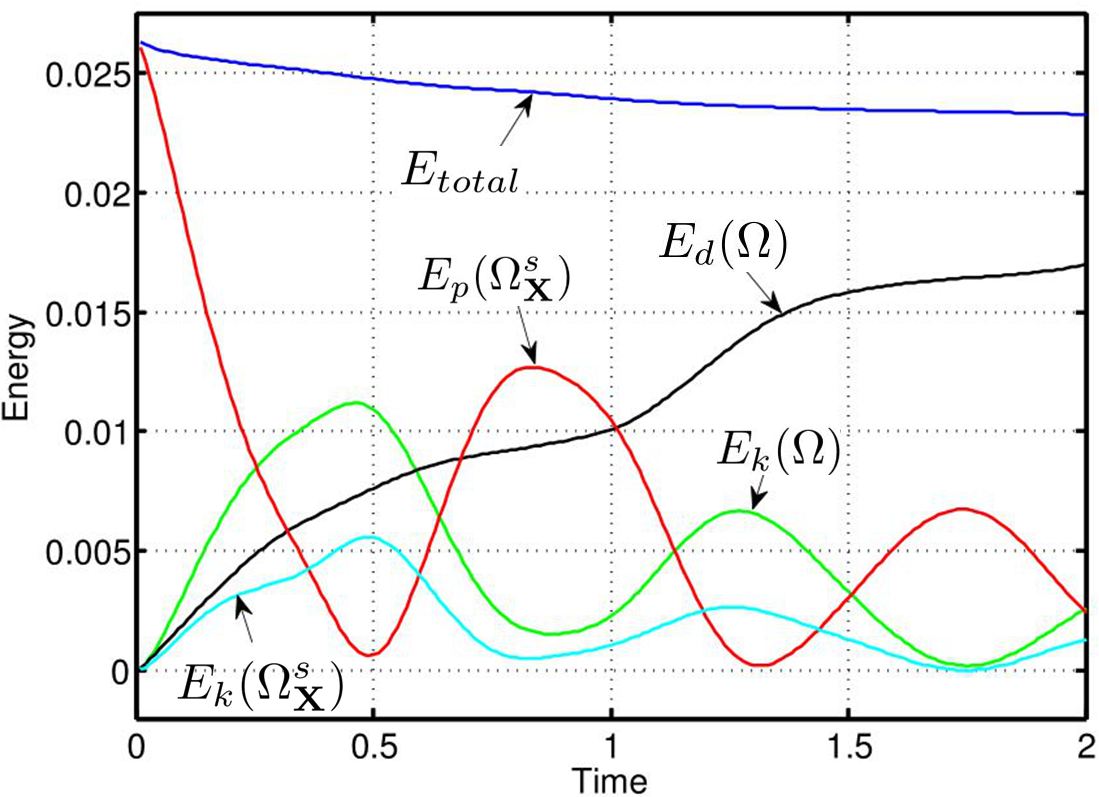}	    		
	\captionsetup{justification=centering}
	\caption {\scriptsize Evolution of energy, $\Delta t=5.0\times 10^{-3}$.} 
	\label{energy_evolution_stretch}
\end{figure}

\begin{figure}[h!]
	\begin{minipage}[t]{0.5\linewidth}
		\centering  
		\includegraphics[width=2.0in,angle=0]{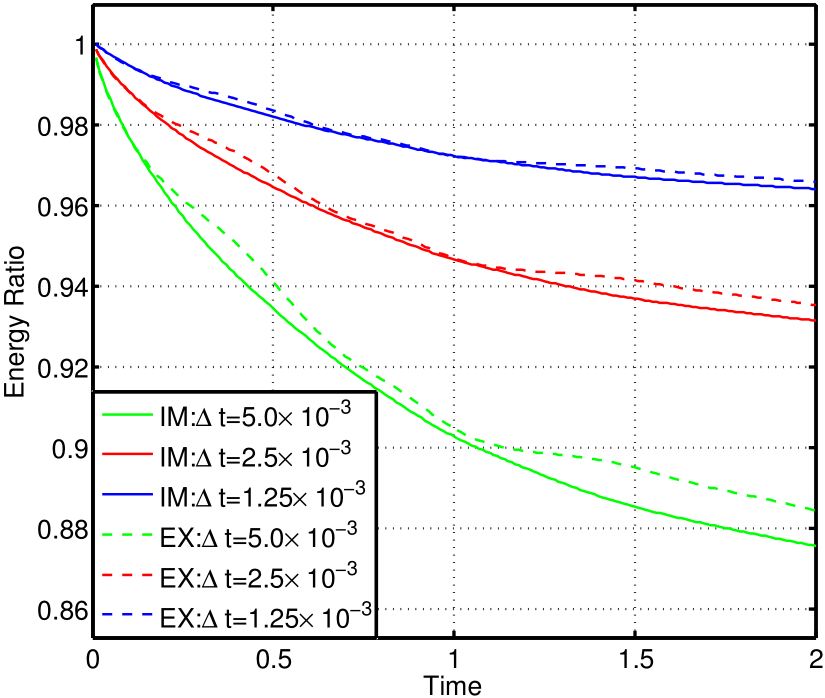}	
		\captionsetup{justification=centering}
		\caption*{\scriptsize(a) Energy ratio against time (defined in (\ref{energy estimate_after_time_discretization_closed})),}
	\end{minipage}
	\begin{minipage}[t]{0.5\linewidth}
		\centering  
		\includegraphics[width=1.9in,angle=0]{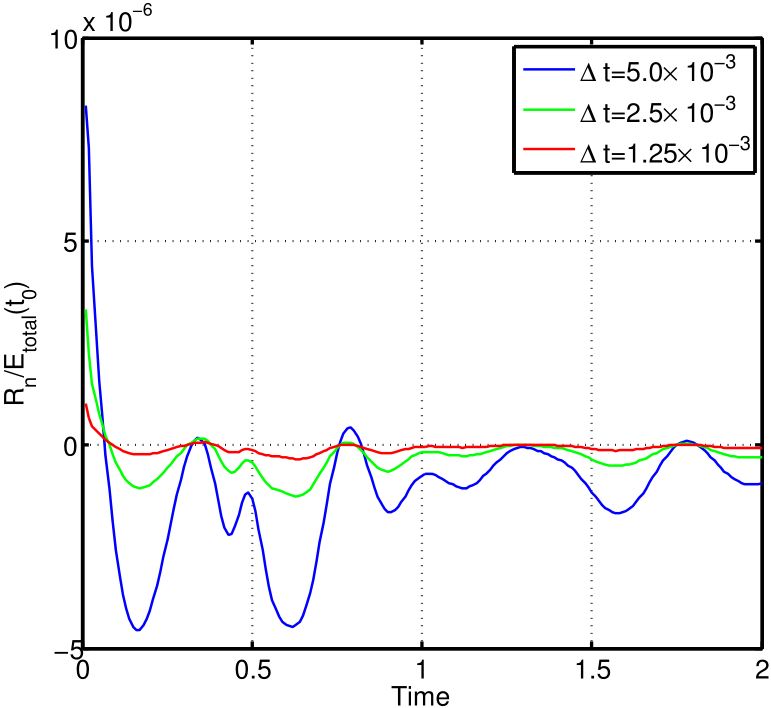}
	\caption*{\scriptsize (b) $R_n/E_{total}\left(t_0\right)$ against time (defined in (\ref{residual})).} 
	\end{minipage}  
	\captionsetup{justification=centering}
	\caption {\scriptsize Evolution of the energy ratio and residual $R_n$ for the test problem of stretched disc.} 
	\label{time_mesh_conservation_stretch}	   		
\end{figure}

\subsection{Oscillating ball driven by an initial kinetic energy}
\label{subsec:ball}
In this section, we consider a 3D oscillating ball, which is an extension of the example in section \ref{subsec:oddbaike}. The ball is initially located at the center of $\Omega=[0,1]\times[0,1]\times[0,0.6]$ with a radius of $0.2$. Using the property of symmetry this computation is carried out on $1/8$ of domain $\Omega$: $[0,0.5]\times[0,0.5]\times[0,0.3]$. The initial velocities of $x$ and $y$ components are the same as that used in section \ref{subsec:oddbaike} and the $z$ component is set to be 0 at the beginning. We adopt the same parameter and mesh size defined in section \ref{subsec:oddbaike} (with the same mesh size in the z direction). A snapshot of the $1/8$ solid ball and the corresponding fluid velocity norm are presented in Figure \ref{Velocity}, and the nonincreasing energy property is presented in Figure \ref{time_mesh_conservation_ball}.

\begin{figure}[h!]
	\begin{minipage}[t]{0.5\linewidth}
		\centering  
		\includegraphics[width=2.0in,angle=0]{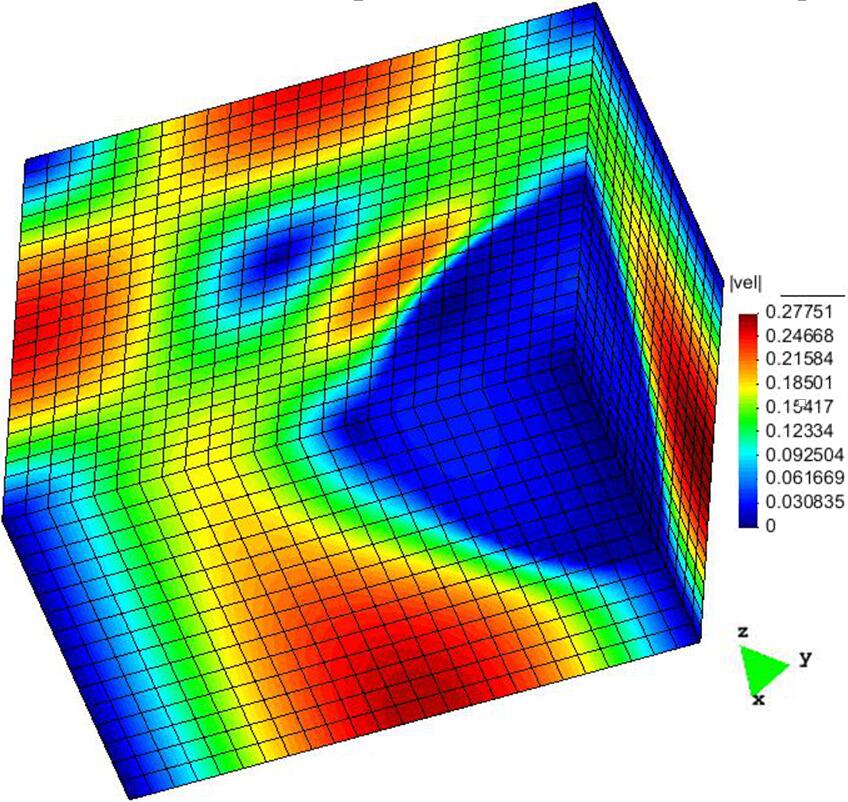}	
		\caption*{\scriptsize(a) Fluid mesh,}
	\end{minipage}
	\begin{minipage}[t]{0.5\linewidth}
		\centering  
		\includegraphics[width=1.5in,angle=0]{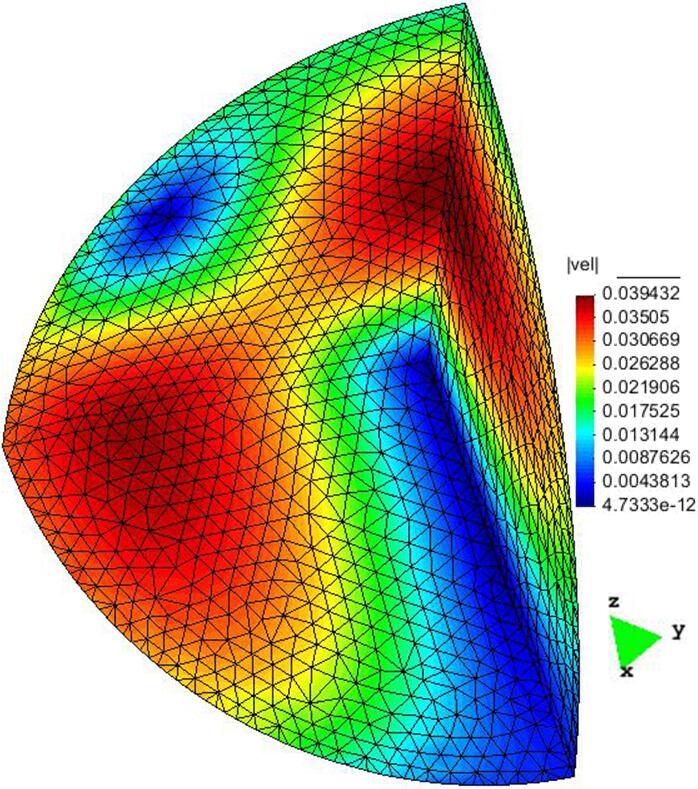}
		\caption*{\scriptsize(b) solid mesh.}
	\end{minipage}   		
	\captionsetup{justification=centering}
	\caption {\scriptsize Velocity norm at $t=0.2$.} 
	\label{Velocity}
\end{figure}

\begin{figure}[h!]
\centering  
\includegraphics[width=1.9in,angle=0]{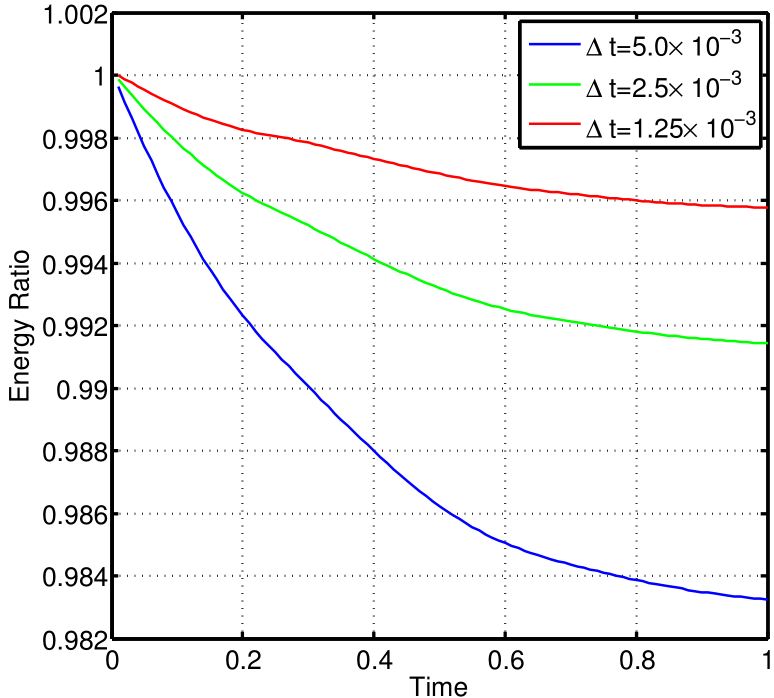}
\captionsetup{justification=centering}
\caption {\scriptsize Evolution of the energy ratio (defined in (\ref{energy estimate_after_time_discretization_closed})) for the test problem of oscillating ball.} 
\label{time_mesh_conservation_ball}	   		
\end{figure}

%%%%%%%%%%%%%%%%%%%%%%%%%%%%%%%%%%%%%%%%%%
\section{Conclusions}
\label{sec:conclusions}
In this article, we first introduce an implicit version of \cite{Wang_2017} for the one-field fictitious domain method (one-field FDM) based upon updating the solid deformation tensor ${\bf F}$. Then the energy-preserving property for this one-field FDM is proved on the continuous level, and the energy-nonincreasing property is proved after discretization in time and space. The energy property for an explicit scheme is also analyzed in \ref{appendix_explicit}. Finally, a selection of numerical tests are presented to demonstrate this theoretical energy estimate in both two and three dimensions. It has therefore been demonstrated that the proposed one-field FDM is a stable, robust and computationally efficient technique for the solution of a wide range of fluid-structure interaction problems.

\appendix
\section{Stability analysis after space discretization}
As with the previous stability estimate (Proposition \ref{lec_backward_Euler}) after time discretization, we have the following estimate after space discretization. 
\label{appendix_space_dis}
\begin{proposition} \label{lec_backward_Euler_space}
	Let $\left({\bf u}_{n+1}^h, p_{n+1}^h\right)$ be the solution pair of Problem \ref{problem_weak_after_space_discretization}. If $\rho^\delta\ge 0$, then
	\begin{equation}\label{energy_estimate_after_time_discretization_backward_Euler_space}
	\begin{split}
	&\frac{\rho^f}{2}\int_{\Omega^h}\left|{\bf u}_{n+1}^h\right|^2d{\bf x}
	+\frac{\rho^\delta}{2}\int_{\Omega_{\bf X}^{sh}}\left|{\bf u}_{n+1}^{sh}\right|^2d{\bf X}
	+\int_{\Omega_{\bf X}^{sh}}\Psi\left({\bf F}_{n+1}^{sh}\right)d{\bf X}  \\
	&+\frac{\Delta t\mu^f}{2}\sum_{k=1}^{n+1}\int_{\Omega^h}{\rm D}{\bf u}_k^h:{\rm D}{\bf u}_k^hd{\bf x}\\
	&\le \frac{\rho^f}{2}\int_{\Omega^h}\left|{\bf u}_n^h\right|^2d{\bf x}
	+\frac{\rho^\delta}{2}\int_{\Omega_{\bf X}^{sh}}\left|{\bf u}_n^h\right|^2d{\bf X}
	+\int_{\Omega_{\bf X}^{sh}}\Psi\left({\bf F}_n^{sh}\right)d{\bf X}\\
	&+\frac{\Delta t\mu^f}{2}\sum_{k=1}^n\int_{\Omega^h}{\rm D}{\bf u}_k^h:{\rm D}{\bf u}_k^hd{\bf x}+R_{n+1}^h,
	\end{split}
	\end{equation}
where
\begin{equation}\label{residual_space}
R_{n+1}^h=\frac{c_1\Delta t^2}{2}\int_{\Omega_{\bf X}^{sh}}\left(\left|\left({\bf F}_{n+1}^{sh}\right)^{-1}\nabla_{\bf X}{\bf u}_{n+1}^{sh}\right|^2
-\left|\nabla_{\bf X}{\bf u}_{n+1}^{sh}\right|^2\right)d{\bf X}.
\end{equation}
\end{proposition}
\begin{proof}
	Let ${\bf v}={\bf u}_{n+1}^h$ in (\ref{weak_form1_discretization_space}) and multiply $\Delta t$ on both side of the equation, and then let $q=p_{n+1}^h$ in (\ref{weak_form2_discretization_space}) and substitute into equation (\ref{weak_form1_discretization_space}), we get:	
	\begin{equation}\label{one_equation_backward_Euler_space}
	\begin{split}
	&\rho^f\int_{\Omega^h}\left({\bf u}_{n+1}^h-{\bf u}_n^h\right)\cdot{\bf u}_{n+1}^hd{\bf x}
	+\frac{\Delta t\mu^f}{2}\int_{\Omega^h}{\rm D}{\bf u}_{n+1}^h:{\rm D}{\bf u}_{n+1}^hd{\bf x}\\
	&+\rho^{\delta}\int_{\Omega_{\bf X}^{sh}}\left({\bf u}_{n+1}^{sh}-{\bf u}_n^{sh}\right)\cdot{\bf u}_{n+1}^{sh}d{\bf X}\\
	&+c_1\Delta t\int_{\Omega_{\bf X}^{sh}}{\bf F}_{n+1}^{sh}:\nabla_{\bf X}{\bf u}_{n+1}^{sh}d{\bf X}
	-c_1\Delta t\int_{\Omega_{n+1}^{sh}}\nabla\cdot{\bf u}_{n+1}^{sh}d{\bf x}
	=0.
	\end{split}
	\end{equation}
	Using the Cauchy-Schwarz inequality and the fact $ab\le\frac{a^2+b^2}{2}$, we have:
	\begin{equation*}
	\int_{\omega}{\bf u}_n\cdot{\bf u}_{n+1}d{\bf x}
	\le \left\|{\bf u}_n\right\|_{0,\omega} \left\|{\bf u}_{n+1}\right\|_{0,\omega}
	\le \frac{\|{\bf u}_n\|_{0,\omega}^2+\|{\bf u}_{n+1}\|_{0,\omega}^2}{2},
	\end{equation*}
	where $\omega=\Omega^h$ or $\Omega_{n+1}^{sh}$. Notice that Lemma \ref{energy_estimate} to \ref{lemma_convection_zero_discretization_in_time} still hold after space discretization, then substituting the above relation into (\ref{one_equation_backward_Euler_space}) gives (\ref{energy_estimate_after_time_discretization_backward_Euler_space}).
\end{proof}

\section{Energy estimate for a two-step explicit splitting scheme} 
\label{appendix_explicit}
In this section, we analyze the energy property for the 2-step explicit splitting scheme introduced in \cite{Wang_2017}, which can be stated as follows (corresponding to the implicit Problem \ref{problem_weak_after_space_discretization}):

\begin{problem}\label{problem_weak_after_space_discretization_explicit}	
Given ${\bf u}_n^h$, $p_n^h$ and $\Omega_n^{sh}$, find ${\bf u}_{n+1}^h\in V^h(\Omega^h)^d$, $p_{n+1}^h \in L^h(\Omega^h)$ and $\Omega_{n+1}^{sh}$, such that for $\forall{\bf v}\in V^h(\Omega^h)^d$, $\forall q\in L^h(\Omega^h)$, the following 5 relations hold:

(1) convetion step:	
	\begin{equation}\label{convection}
	\rho^f\int_{\Omega^h}\frac{{\bf u}_{n+1/2}^h-{\bf u}_n^h}{\Delta t} \cdot{\bf v}d{\bf x}
	+\rho^f \int_{\Omega^h}\left({\bf u}_{n+1/2}^h\cdot\nabla\right){\bf u}_{n+1/2}^h\cdot{\bf v}d{\bf x}=0,
	\end{equation}
	
(2) diffusion step:
	\begin{equation}\label{diffusion}
	\begin{split}
	&\rho^f\int_{\Omega^h}\frac{{\bf u}_{n+1}^h-{\bf u}_{n+1/2}^h}{\Delta t} \cdot{\bf v}d{\bf x}
	+\frac{\mu^f}{2}\int_{\Omega^h}{\rm D}{\bf u}_{n+1}^h:{\rm D}{\bf v}d{\bf x}\\
	&-\int_{\Omega^h}p_{n+1}^h\nabla \cdot {\bf v}d{\bf x}
	+\rho^{\delta}\int_{\Omega_{\bf X}^{sh}}\frac{{\bf u}_{n+1}^{sh}-{\bf u}_n^{sh}}{\Delta t} \cdot{\bf v}^sd{\bf X}\\
	&+c_1\Delta t\int_{\Omega_{\bf X}^{sh}}\nabla_{\bf X}{\bf u}_{n+1}^{sh}:\nabla_{\bf X}{\bf v}^sd{\bf X}
	-c_1\int_{\Omega_n^{sh}}J_n^{-1}\nabla_n\cdot{\bf v}^sd{\bf x}\\
	&=-c_1\int_{\Omega_{\bf X}^{sh}}{\bf F}_n^{sh}:\nabla_{\bf X}{\bf v}^sd{\bf X},
	\end{split}
	\end{equation}
	\begin{equation}\label{weak_form2_discretization_space_ex}
	-\int_{\Omega} q\nabla \cdot {\bf u}_{n+1}^hd{\bf x}=0,
	\end{equation}	
	\begin{equation}\label{updating_disretizatoin_space_ex}
	\Omega_{n+1}^{sh}=\left\{{\bf x}:{\bf x}={\bf x}_n+\Delta t{\bf u}_{n+1}^{sh}, {\bf x}_n\in\Omega_n^{sh} \right\},
	\end{equation}
	and
	\begin{equation}\label{update_f_ex}
	{\bf F}_{n+1}^{sh}={\bf F}_n^{sh}+\Delta t\nabla_{\bf X}{\bf u}_{n+1}^{sh},
	\end{equation}	
	where ${\bf u}^{sh}=P_{n+1}\left({\bf u}^h\right)$, ${\bf v}^s=P_{n+1}\left({\bf v}\right)$, and $\nabla_n(\cdot)=\frac{\partial(\cdot)}{\partial{\bf x}_n}$.
\end{problem}
As with the previous analysis for the implicit scheme, if we let ${\bf v}={\bf u}_{n+1}^h$ in equations (\ref{convection}), (\ref{diffusion}) and (\ref{weak_form2_discretization_space_ex}), adding up these three equations, using (\ref{update_f_ex}) and ${\bf u}^{sh}=P_{n+1}\left({\bf u}^h\right)$, gives the energy estimate as follows.
\begin{proposition}
	Let $\left({\bf u}_{n+1}, p_{n+1}\right)$ be the solution pair of Problem \ref{problem_weak_after_space_discretization_explicit}. If $\rho^\delta\ge 0$, then
	\begin{equation}
	\begin{split}
	&\frac{\rho^f}{2}\int_{\Omega^h}\left|{\bf u}_{n+1}^h\right|^2d{\bf x}
	+\frac{\rho^\delta}{2}\int_{\Omega_{\bf X}^{sh}}\left|{\bf u}_{n+1}^h\right|^2d{\bf X}
	+\int_{\Omega_{\bf X}^{sh}}\Psi\left({\bf F}_{n+1}^{sh}\right)d{\bf X}  \\
	&+\frac{\Delta t\mu^f}{2}\sum_{k=1}^{n+1}\int_{\Omega^h}{\rm D}{\bf u}_k^h:{\rm D}{\bf u}_k^hd{\bf x}\\
	&\le \frac{\rho^f}{2}\int_{\Omega^h}\left|{\bf u}_n^h\right|^2d{\bf x}
	+\frac{\rho^\delta}{2}\int_{\Omega_{\bf X}^{sh}}\left|{\bf u}_n^h\right|^2d{\bf X}
	+\int_{\Omega_{\bf X}^{sh}}\Psi\left({\bf F}_n^{sh}\right)d{\bf X}\\
	&+\frac{\Delta t\mu^f}{2}\sum_{k=1}^n\int_{\Omega^h}{\rm D}{\bf u}_k^h:{\rm D}{\bf u}_k^hd{\bf x}+R_{n+1}^{im}+R_{n+1}^{ex}+R_{n+1}^{split},
	\end{split}
	\end{equation}
where 
$
R_{n+1}^{im}=R_{n+1}^h
$
as defined in (\ref{residual_space}).
\begin{equation}
R_{n+1}^{ex}=c_1\Delta t\int_{\Omega_{\bf X}^{sh}}\left(\nabla_n\cdot{\bf u}_{n+1}^{sh}-\nabla\cdot{\bf u}_{n+1}^{sh}\right)d{\bf X},
\end{equation}
and
\begin{equation}
R_{n+1}^{split}=-\Delta t\rho^f \int_{\Omega^h}\left({\bf u}_{n+1/2}^h\cdot\nabla\right){\bf u}_{n+1/2}^h\cdot{\bf u}_{n+1}^hd{\bf x}.
\end{equation}
\end{proposition}

%\bibliography{mybibfile}
\end{document}